\documentclass[12pt]{article}
\usepackage{amsmath,amssymb,amsthm}
\usepackage[square,comma,numbers,sort&compress]{natbib}
\usepackage[all]{xy}
\newtheorem{thm}{Theorem}
\newtheorem{defi}{Definition}
\newtheorem{coro}{Corollary}
\newtheorem{prop}{Proposition}
\newtheorem{lemma}{Lemma}

\theoremstyle{remark}

\newtheorem{rem}{Remark}

\def\Le{\mathcal L}

\def\Ce{\mathcal C}
\def\Ee{\mathcal E}\def\Fe{\mathcal F}
\def\Ha{\mathcal H}

\def\Ka{\mathcal K}
\def\Te{\mathcal T}
\def\Me{\mathcal M}

\def\states{\mathfrak S}
\def\Se{\mathcal S}
\def\Pe{\mathcal P}

\def\De{\mathcal D}
\def\Tr{\mathrm{Tr}\,}

\def\<{\langle}
\def\>{\rangle}
\begin{document}
\title{Comparison of  quantum  channels and statistical experiments}
\author{Anna Jen\v cov\'a\footnote{jenca@mat.savba.sk}\\ Mathematical Institute, Slovak Academy of Sciences\\
\v Stef\'anikova 49, Bratislava, Slovakia}
\date{}
\maketitle

\abstract{For a pair of quantum channels with the same input space, we show that the possibility of approximation of one channel by post-processings  of the other channel can be characterized by comparing the success probabilities for the two ensembles obtained as outputs for any ensemble on the input space coupled with an ancilla. This provides an operational interpretation to a natural extension of Le Cam's deficiency to quantum channels. In particular, we obtain a version of the randomization criterion for quantum statistical experiments. The proofs are based on some properties of the diamond norm and its dual, which are of independent interest.  }

\section{Introduction}

The classical randomization criterion \cite{lecam1964sufficiency,strasser1985statistics,torgersen1991comparison} for statistical experiments is an important result of statistical decision theory.  It makes a link between the performance of decision rules available for corresponding decision problems and the possibility of approximating one experiment by randomizations of the other. A special case is the Blackwell-Sherman-Stein (BSS) theorem \cite{blackwell1951comparison, sherman1951theorem, stein1951comparison}, which states that a stochastic mapping (Markov kernel) transforming one experiment into the other exists if and only if the optimal decision rules for the former experiment have smaller risks.

 A quantum statistical experiment is a parametrized  family  of density operators and the role of  stochastic mappings is played  by completely positive trace preserving maps, or {channels}. Decision rules for (classical) decision spaces are given by positive operator valued measures (POVMs). The corresponding comparison of quantum experiments was studied e.g. in \cite{matsumoto2010randomization, jencova2012comparison}, but in this case the BSS theorem does not hold, even if one of the experiments is classical, see \cite{matsumoto2014anexample}.  

 A quantum  version of the BSS theorem was first obtained by 
Shmaya \cite{shmaya2005comparison} in the framework of the so-called quantum information structures. In \cite{buscemi2012comparison}, a theory of comparison for both classical and quantum experiments is developed in terms of  statistical morphisms and a general form of BSS theorem is proved. In both works, either additional entanglement or composition of the experiment with a complete set of states is required. On the other hand, Matsumoto \cite{matsumoto2010randomization} introduced a natural generalization of classical decision problems to quantum ones and proved  a quantum randomization criterion in this setting, using the minimax theorem similarly as in the classical case (see e.g. \cite{strasser1985statistics}). The main drawback of this approach is the lack of operational interpretation for quantum decision problems.

Comparison of channels is an extension of the theory of comparison of experiments. 
A natural idea is the following: given channels $\Phi$ and $\Psi$ with the same input space, compare the two experiments obtained as outputs  for any given input experiment. If the output of, say, $\Psi$ is always more informative for any 	decision problem, we might say that $\Psi$ is less noisy than $\Phi$. In the classical case, an ordering of channels was first introduced in the work by Shannon \cite{shannon1958anote}, where a coding/decoding criterion was applied. Other  orderings were studied e.g. in \cite{koma1977comparison, czko1981information},  see \cite{raginsky2011shannon, buscemi2015degradable} for some more recent works.

In the quantum setting, it is possible to introduce a stronger condition than the one described above, namely to use experiments on the input space coupled with any ancilla. As it turns out, $\Psi$ is less noisy in this stronger sense if and only if $\Phi$ is a post-processing of $\Psi$, (or $\Psi$ is degradable into $\Phi$, in the terminology of information theory) which means that there exists some channel $\alpha$ such that $\Phi=\alpha\circ \Psi$. In fact, it is enough to compare guessing probabilities for ensembles of states. This remarkable result was first obtained by Chefles in \cite{chefles2009quantum}, using the results of \cite{shmaya2005comparison}. It was  extended and refined in \cite{buscemi2012comparison}, in particular it was proved that no entanglement in the input ensemble is needed. Some applications were already found in \cite{budast2014game, buscemi2014fully, buda2016equivalence, buscemi2015degradable}.

The aim of the present work is to establish an approximate version, which may be called the randomization criterion for quantum channels. More precisely, we study an extension of Le Cam's deficiency for quantum channels, defined as the precision, measured in the diamond norm, up to which one can approximate one channel by post-processings of the other. We show that this notion of deficiency can be characterized by comparing success probabilities for output ensembles, with respect to the success probability of the input ensemble. These results are then applied to statistical experiments and a quantum randomization criterion is proved in terms of success probabilities. 
 We also discuss the case when no ancilla is present and show that this leads to the  classical deficiency of experiments, studied e.g. in \cite{matsumoto2010randomization,jencova2012comparison}.

The diamond norm naturally appears as a distinguishability norm for quantum channels \cite{kitaev1997quantum, watrous2005notes}. As it was observed in \cite{jencova2014base}, this norm can be defined using only the order structure given by the cone of completely positive maps and has a similar relation to the set of channels as the trace norm has to the set of states. We also show that the dual norm on positive elements can be expressed as the optimal success probability for a certain ensemble. These properties provide a very convenient framework for proving  our results and are of independent interest.

\section{Notations and preliminaries}

Throughout the paper, all Hilbert spaces  are finite dimensional.
If $\Ha$ is a Hilbert space, we denote  $d_\Ha:=\dim(\Ha)$ and fix an orthonormal basis $\{|e^\Ha_i\>, i_H=1\dots d_\Ha\}$  in $\Ha$. 
We will denote the algebra of linear operators on $\Ha$ by  $B(\Ha)$, the set of positive operators in $B(\Ha)$ by $B(\Ha)^+$ and 
 the real vector space of Hermitian elements in $B(\Ha)$ by $B_h(\Ha)$. 
The set of states, or density operators, on $\Ha$ will be denoted by 
\[
\states(\Ha):=\{\sigma\in B(\Ha)^+,\ \Tr\sigma=1\}.
\]
\subsection{Spaces of Hermitian maps}

Let $\Le(\Ha,\Ka)$ denote the  space of real linear maps $B_h(\Ha)\to B_h(\Ka)$. Then $\Le(\Ha,\Ka)$ can be identified with the space of Hermitian linear maps $B(\Ha)\to B(\Ka)$. The set $\Le(\Ha,\Ka)^+$  of completely positive maps forms  a closed convex cone in $\Le(\Ha,\Ka)$ which is pointed 
and generating. With this cone, $\Le(\Ha,\Ka)$ becomes an ordered vector space. We will denote the corresponding  order by $\le$.  An element of $\Le(\Ha,\Ka)^+$ that preserves trace is usually called a channel. We will denote the set of all channels by $\Ce(\Ha,\Ka)$.

For $\phi\in \Le(\Ha,\Ka)$, we will denote by $\phi^*$ its adjoint with respect to the Hilbert-Schmidt inner product. That is, $\phi^*\in \Le(\Ka,\Ha)$ is defined by
\[
\Tr [\phi^*(A)B]=\Tr [A\phi(B)],\qquad A\in B_h(\Ka),\ B\in B_h(\Ha).
\]
Note that $\phi^*$ is completely positive if and only if $\phi$ is, moreover, $\phi$ is a channel if and only if $\phi^*$ is completely positive and unital, $\phi^*(I)=I$.

A special class of elements in $\Le(\Ha,\Ka)$ are the classical-to-quantum (cq-)maps and quantum-to-classical (qc-)maps, defined as follows. Let  $A=\{A_1,\dots,A_{d_\Ha}\}$  be any collection of operators in $B_h(\Ka)$ and define
\[
\phi^{cq}_A: X\mapsto \sum_i \<e^\Ha_i,Xe^\Ha_i\>A_i,\qquad X\in B_h(\Ha). 
\]
Similarly, for $B=\{B_1,\dots, B_{d_\Ka}\}$, $B_i\in B_h(\Ha)$, we put 
\[
\phi^{qc}_B: X\mapsto \sum_i \Tr[XB_i]|e^\Ka_i\>\<e_i^\Ka|,\qquad X\in B_h(\Ha).
\]
It is easy to see that $\phi^{cq}_A\in \Le(\Ha,\Ka)$ and it is completely positive if and only if each $A_i \in B(\Ka)^+$, similarly for $\phi^{qc}_B$. Moreover, $\phi^{cq}_A$ is a channel if and only if $A_i\in \states(\Ka)$ for all $i$ and  $\phi^{qc}_B$ is a channel if and only if $B$ is a collection of positive operators such that $\sum_i B_i=I$. Such a collection is called a POVM (positive operator valued measure)  and is used in the description of measurements with values in the set $\{1,\dots,d_\Ka\}$. We will denote the set of all $n$-valued POVMs on $\Ha$ by $\Me(\Ha,n)$.

\begin{rem}\label{rem:cqqc}
If $A=\{A_1,\dots,A_n\}$ is a collection of operators in $B_h(\Ha)$, then $\phi_A^{cq}$, resp. $\phi_A^{qc}$, denotes the cq-map in $\Le(\mathbb C^n,\Ha)$, resp. the qc-map in $\Le(\Ha,\mathbb C^n)$, defined by $A$ and the standard basis in $\mathbb C^n$. In this way, any POVM  can be identified with a qc-channel and any finite collection of states on $\Ha$ with a cq-channel. Moreover, if  $\phi\in \Le(\Ha,\Ka)$, then it is easy to see that $\phi\circ\phi^{cq}_A=
\phi^{cq}_{\phi(A)}$, where $\phi(A)=\{\phi(A_1),\dots,\phi(A_n)\}$. Similarly, if $\psi\in \Le(\Ka,\Ha)$,  $\phi^{qc}_A\circ\psi=\phi^{qc}_{\psi^*(A)}$.
If $F\in \Me(\Ha,n)$ and $\psi$ is a channel, then $\psi^*(F)\in \Me(\Ka,n)$  is called the pre-processing of $F$ by $\psi$.

\end{rem}

 For $\phi\in \Le(\Ha,\Ha)$, we  define 
\[
s(\phi)=\sum_{i,j} \<e^\Ha_i,\phi(|e^\Ha_i\>\<e^\Ha_j|)e^\Ha_j\>.
\] 
It is easy to see that $s$ defines a  linear functional $s: \Le(\Ha,\Ha)\to \mathbb R$. The next lemma shows that this functional has tracelike properties with respect to composition of maps.
 
\begin{lemma}\label{lemma:s}
For all $\phi\in \Le(\Ha,\Ka)$, $\psi\in \Le(\Ka,\Ha)$, $s(\psi\circ\phi)=s(\phi\circ\psi)$.
\end{lemma}

\begin{proof} 
Let $\phi_{B,A}$ denote the map $B(\Ka)\ni X\mapsto \Tr [BX] A$, with $A\in B_h(\Ha)$, $B\in B_h(\Ka)$. We have
 \begin{align*}
s(\phi_{B,A}\circ\phi)&=\sum_{i,j} \<e^\Ha_i,Ae^\Ha_j\>\Tr [B\phi(|e^\Ha_i\>\<e^\Ha_j|)]=\sum_{i,j}\<e^\Ha_i,Ae^\Ha_j\> \<e^\Ha_j,\phi^*(B)e^\Ha_i\>\\&=\Tr 
[\phi(A)B].
 \end{align*}
  Similarly,
\begin{align*}
s(\phi\circ\phi_{B,A})&=\sum_{i,j}\<e^\Ka_i,\phi(A)e^\Ka_j\> \Tr [B|e^\Ka_i\>\<e^\Ka_j|]=\sum_{i,j}\<e^\Ka_i,\phi(A)e^\Ka_j\> \<e^\Ka_j, Be^\Ka_i\>\\
&=\Tr[\phi(A)B]=s(\phi_{B,A}\circ\phi).
\end{align*}
Since the maps $\phi_{B,A}$ generate $\Le(\Ka,\Ha)$ and $s$ is linear, the statement follows.

\end{proof}

We now identify the dual space of $\Le(\Ha,\Ka)$ with $\Le(\Ka,\Ha)$, where duality is given by
\[
\<\psi,\phi\>=s(\psi\circ\phi),\qquad \phi\in \Le(\Ha,\Ka),\ \psi\in\Le(\Ka,\Ha).
\]
This duality is closely related to the inner product $\<\cdot,\cdot\>''$ in $\Le(\Ha,\Ka)$, introduced in \cite{skowronek2011cones}.
Note that the properties of $s$ imply that we have $\<\phi,\psi\>=\<\psi,\phi\>$ and 
\[
\<\phi,\xi\circ\psi\>=\<\phi\circ\xi,\psi\>=\<\psi\circ\phi,\xi\>
\]
whenever $\phi$, $\psi$ and $\xi$ are maps with appropriate input and output spaces.
The  dual cone of positive functionals satisfies
\[
(\Le(\Ha,\Ka)^+)^*:=\{\psi\in \Le(\Ka,\Ha), \<\psi,\phi\>\ge 0, \forall \phi\in \Le(\Ha,\Ka)^+\}=\Le(\Ka,\Ha)^+
\]
so that, in this sense, the cone of completely positive maps is self-dual. 

\begin{rem}\label{rem:choi}
Let us denote 
\[
X_\Ha:=\sum_{i,j} |e^\Ha_i\>\<e^\Ha_j|\otimes |e^\Ha_i\>\<e^\Ha_j|.
\]
  The {Choi representation} (\cite{choi1975completely})
$C: \phi\mapsto (\phi\otimes id_\Ha)(X_\Ha)$ provides an order-isomorphism of $\Le(\Ha,\Ka)$ onto $B_h(\Ka\otimes \Ha)$ with the cone of positive operators $B(\Ka\otimes\Ha)^+$.  Note  also that for any $\phi\in \Le(\Ha,\Ha)$, $s(\phi)=\Tr C(\phi)X_\Ha$, so that for $\phi\in  \Le(\Ha,\Ka)$, 
$\psi\in \Le(\Ka,\Ha)$ we obtain
\begin{align*}
\<\psi,\phi\>&=s(\psi\circ\phi)=\Tr [C(\psi\circ\phi)X_\Ha]=\Tr [(\psi\otimes id_\Ha)(C(\phi))X_\Ha]\notag \\
&=\Tr [C(\phi)(\psi^*\otimes id_\Ha)(X_\Ha)]=\Tr [C(\phi)C(\psi^*)].
\end{align*}
It follows that $\psi$, seen as a linear functional on $\Le(\Ha,\Ka)$, is identified with the functional on $B_h(\Ka\otimes \Ha)$ defined by $C(\psi^*)$ through the natural duality of $B_h(\Ka,\otimes \Ha)$ with itself, given by the trace.
It is of course possible to use this representation and we will do it in some places, but for our purposes it is mostly  more convenient to work with the spaces of mappings.

\end{rem}

\subsection{The diamond norm and its dual}

The diamond norm in $\Le(\Ha,\Ka)$ is defined by 
\begin{align}\label{eq:diamond}
\|\phi\|_\diamond &=\sup_{ d_{\Ha_0}<\infty}\,\sup_{\rho\in \states(\Ha\otimes \Ha_0)}\|(\phi\otimes id)(\rho)\|_1\\
&=\sup_{\rho\in \states(\Ha\otimes \Ha)}\|(\phi\otimes id)(\rho)\|_1,\end{align}
where $\|\cdot\|_1$ denotes the $L_1$-norm, or trace norm, in $B(\Ka\otimes \Ha)$. It was proved in 
 \cite{jencova2014base} that the diamond norm is related to the order structure in $\Le(\Ha,\Ka)$ and the set $\Ce(\Ha,\Ka)$, similarly as the trace norm is related to the set of states. Namely, it was shown that
\begin{equation}\label{eq:diamondinf}
\|\phi\|_\diamond = \inf_{\alpha\in \Ce(\Ha,\Ka)} \inf\{\lambda>0, -\lambda\alpha\le \phi\le \lambda\alpha\},\quad \phi\in \Le(\Ha,\Ka).
\end{equation}
It was also shown that the dual norm in $\Le(\Ka,\Ha)$, which we will denote by $\|\cdot\|^\diamond$, has a similar relation to the set of erasure channels $\{\phi_\sigma:B(\Ka)\to B(\Ha),\ A\mapsto \Tr[A]\sigma,\ \sigma\in \states(\Ha)\}$, so that 
\begin{equation}\label{eq:updiamondinf}
\|\psi\|^\diamond =\inf_{\sigma\in \states(\Ha)} \inf \{\lambda>0, -\lambda\phi_\sigma \le \psi\le \lambda\phi_\sigma\},\quad \psi\in \Le(\Ka,\Ha).
\end{equation}
We list some useful properties of these norms.

\begin{prop}\label{prop:properties} \begin{enumerate}
\item[(i)] If $\phi\in \Le(\Ha,\Ka)^+$, then 
\[
\|\phi\|_\diamond=\sup_{\sigma\in \states(\Ha)} \Tr[\phi(\sigma)],\quad
\|\phi\|^\diamond =\sup_{\alpha\in \Ce(\Ka,\Ha)} \<\alpha,\phi\>.
\]
\item[(ii)] If $\phi,\psi\in \Ce(\Ha,\Ka)$, then 
\[
\|\phi-\psi\|_\diamond =2\sup_{\gamma \ge 0, \|\gamma\|^\diamond\le 1} \<\gamma,\phi-\psi\>
\]
\item[(iii)] If  $\chi\in \Ce(\Ka,\Ka')$ and $\xi\in\Ce(\Ha',\Ha)$, then the maps $\phi\mapsto \chi\circ\phi$ and $\phi\mapsto \phi\circ\xi$ are contractions with respect to both $\|\cdot\|_\diamond$ and $\|\cdot\|^\diamond$.
\end{enumerate}

\end{prop}

\begin{proof}
It is clear from (\ref{eq:diamondinf}) that for $\psi\in \Le(\Ka,\Ha)$, $\|\psi\|_\diamond\le 1$ if and only if there is some channel $\alpha$ such that $-\alpha\le \psi\le \alpha$.
Since $\|\cdot\|^\diamond$ is the dual of $\|\cdot\|_\diamond$ and $\phi\ge 0$, we have
\[
\|\phi\|^\diamond=\sup_{\|\psi\|_\diamond\le1}\<\psi,\phi\>=\sup_{\alpha\in \Ce(\Ka,\Ha)}\sup_{-\alpha\le \psi\le \alpha} \<\psi,\phi\>\le \sup_{\alpha\in \Ce(\Ka,\Ha)}\<\alpha,\phi\>\le \|\phi\|^\diamond.
\]
 Similarly, we have by Remark \ref{rem:cqqc},
\begin{align*}
\|\phi\|_\diamond&=\sup_{\sigma\in \states(\Ha)}\<\phi_\sigma,\phi\>=\sup_{\sigma\in \states(\Ha)} s(\phi\circ\phi_\sigma)=\sup_{\sigma\in \states(\Ha)}s(\phi_{\phi(\sigma)})\\ &=\sup_{\sigma\in \states(\Ha)}\Tr[\phi(\sigma)].
\end{align*}
This proves (i). For (ii), let $\gamma\in \Le(\Ha,\Ka)^+$, then $\|\gamma\|^\diamond\le 1$ if and only if $\gamma\le\phi_\sigma$ for some $\sigma\in \states(\Ha)$. Put $\beta=2(\gamma-\tfrac12\phi_\sigma)$, then $-\phi_\sigma\le \beta\le \phi_\sigma$, so that $\|\beta\|^\diamond\le 1$ and since $\<\phi_\sigma,\phi-\psi\>=\Tr [\phi(\sigma)-\psi(\sigma)] =0$, we have
\[
\<\gamma,\phi-\psi\>=\frac12\<\beta,\phi-\psi\>\le\frac12\|\phi-\psi\|_\diamond.
\]
Conversely, for any $\beta\in \Le(\Ka,\Ha)$ with  $-\phi_\sigma\le\beta\le\phi_\sigma$ we have $0\le \frac12(\beta+\phi_\sigma)\le\phi_\sigma$. This gives the opposite inequality.

To prove (iii), we may assume that $\|\phi\|_\diamond=1$. Then $-\alpha\le \phi \le \alpha$ for some channel $\alpha$. But $\chi\circ\alpha$ is again a channel and $-\chi\circ\alpha\le \chi\circ\phi\le\chi\circ\alpha$, so that $\|\chi\circ\phi\|_\diamond\le 1$. The case of the dual norm is proved similarly, using the fact that $\chi\circ\phi_\sigma=\phi_{\chi(\sigma)}$ is an erasure channel. 

To prove the same for the second map, observe that the map $ \phi\mapsto \phi\circ\xi$ is the adjoint of the map $\psi\mapsto 
\xi\circ\psi$, this is easily seen from 
\[
\<\phi,\xi\circ\psi\>=s(\phi\circ\xi\circ\psi)=\<\phi\circ\xi,\psi\>.
\]
The statement (iii) now follows by the first part of the proof.

\end{proof} 

We now find more explicit expressions for the norms of cq- and qc-maps. Define $\delta_\Ha\in \Le(\Ha,\Ha)$ by 
\[
\delta_\Ha:=\phi^{cq}_{\{|e_1^\Ha\>\<e_1^\Ha|,\dots,|e_{d_\Ha}^\Ha\>\<e_{d_\Ha}^\Ha| \}}=\phi^{qc}_{\{|e_1^\Ha\>\<e_1^\Ha|,\dots,|e_{d_\Ha}^\Ha\>\<e_{d_\Ha}^\Ha| \}}.
\]
 Clearly, $\delta_{\Ha}\in \Ce(\Ha,\Ha)$ and $\delta_\Ha\circ\delta_\Ha=\delta_{\Ha}$. Moreover, for any $\phi\in \Le(\Ha,\Ka)$, 
 $\phi\circ \delta_\Ha$ is a cq-map and $\phi\circ\delta_\Ha=\phi$ if and only if $\phi$ is a cq-map. In other words, the map $\phi\mapsto \phi\circ\delta_\Ha$ is a positive idempotent map on $\Le(\Ha,\Ka)$ whose range the set of all cq-maps. Similarly, 
 $\psi\mapsto\delta_\Ha\circ\psi$ is a positive idempotent map on $\Le(\Ka,\Ha)$ whose range is the set of all qc-maps.

\begin{lemma}\label{lemma:qccq} Let $A=\{A_1,\dots,A_n\}$ be a collection of operators in $B_h(\Ha)$. Then  
\begin{align*}
\|\phi^{cq}_A\|_{\diamond}&=\max_i \|A_i\|_1 , \qquad 
\|\phi^{cq}_A\|^\diamond=\inf_{\sigma\in \states(\Ha)} \max_i\|\sigma^{-1/2}A_i\sigma^{-1/2}\|,\\
\|\phi^{qc}_A\|^\diamond&=\sum_i \|A_i\|,\qquad 
\|\phi^{qc}_A\|_{\diamond}=\sup_{\sigma\in \states(\Ha)}\sum_i\|\sigma^{1/2}A_i\sigma^{1/2}\|_1.
\end{align*}

\end{lemma}

\begin{proof} We will use the expressions (\ref{eq:diamondinf}) and (\ref{eq:updiamondinf}). Let us consider the first equality.  Assume that $\alpha\in
\Ce(\mathbb C^n,\Ha)$ is  such that $-\lambda\alpha\le \phi^{cq}_A\le \lambda\alpha$, then the same is true for the cq-channel $\alpha\circ\delta_\Ha$. It follows that the infimum in (\ref{eq:diamondinf}) can be taken over cq-channels. Moreover, if $\alpha=\phi^{cq}_\Se$ for some collection $\Se=\{\sigma_i\}$ of states on $\Ha$, then $-\lambda\alpha\le \phi^{cq}_A\le\lambda\alpha$ if and only if $-\lambda \sigma_i\le A_i\le\lambda \sigma_i$ holds for all $i$. It is now enough to note that 
\[
\inf_{\sigma\in \states(\Ha)}\inf \{\lambda>0,-\lambda\sigma\le A_i\le \lambda\sigma\}=\|A_i\|_1.
\]
Further, note that for $\sigma\in\states(\Ha)$,  $-\lambda\phi_\sigma\le\phi^{cq}_A\le\lambda \phi_{\sigma}$ if and only if $-\lambda \sigma\le A_i\le \lambda \sigma$ for all $i$, equivalently,  
 $\mathrm{supp} (A_i)\subseteq \mathrm{supp}(\sigma)$ and 
$\|\sigma^{-1/2}A_i\sigma^{-1/2}\|\le \lambda$ for all $i$.

To prove the second line of equalities, we will use duality of the two norms. Note that  $\phi^{qc}_A\in \Le(\Ha,\mathbb C^n)$ and for any $\psi\in \Le(\mathbb C^n,\Ha)$, 
\[
\<\psi,\phi^{qc}_A\>=\<\psi,\delta_{\mathbb C^n}\circ \phi^{qc}_A\>=\<\psi\circ\delta_{\mathbb C^n},\phi_A^{qc}\>
\]
and $\psi\circ\delta_{\mathbb C^n}=:\phi^{cq}_F$ is a cq-map with $\|\phi^{cq}_F\|_\diamond\le\|\psi\|_\diamond$. Using this and the first part of the proof, we obtain
\begin{align*}
\|\phi^{qc}_A\|^\diamond&=\sup_{\|\phi^{cq}_F\|_\diamond\le 1}\<\phi^{cq}_F,\phi^{qc}_A\>=
\sup_{\max_i\|F_i\|_1\le 1}\sum_i\Tr[F_iA_i]\\ &=\sum_i \sup_{\|F_i\|_1\le 1}\Tr A_i F_i
=\sum_i\|A_i\|.
\end{align*}
Similarly,
\begin{align*}
\|\phi^{qc}_A\|_\diamond&=\sup_{\|\phi^{cq}_F\|^\diamond\le 1}\<\phi^{cq}_F,\phi^{qc}_A\>=\sup_{\sigma\in \states(\Ha)}\sup_{-\sigma\le F_i\le\sigma}\sum_i\Tr[F_iA_i]\\
&=\sup_{\sigma\in \states(\Ha)}\sum_i\|\sigma^{1/2}A_i\sigma^{1/2}\|_1.
\end{align*}

\end{proof}

\subsection{The dual norm and guessing probabilities}

In this paragraph, we relate the dual norm $\|\cdot\|^\diamond$ of completely positive maps to maximal success probabilities in certain multiple hypothesis testing problems.

Let $\Ee=\{\lambda_i,\sigma_i\}_{i=1}^n$ be and ensemble on $\Ha$, that is, $\lambda_i>0$, $\sum_i\lambda_i=1$ are probabilities and
 $\sigma_i\in\states(\Ha)$, $i=1,\dots,n$. In the setting of multiple hypothesis testing, this is interpreted as a set of possible states of a quantum system with prior probabilities and the task is to guess which one is the true state. Any procedure to obtain such a guess can be identified with  some POVM $M\in \Me(\Ha,n)$, where $\Tr \sigma_iM_j$ is interpreted as the probability that $\sigma_j$ is chosen when the true state is $\sigma_i$. The maximal probability of a successful guess is given by 
\[
P_{succ}(\Ee):=\max_{M\in \Me(\Ha,n)} \sum_i\lambda_i\Tr M_i\sigma_i.
\]
  To the ensemble $\Ee$, we assign the cq-map
\begin{equation} \label{eq:ens}
\phi_\Ee:=\phi^{cq}_{\lambda_1\sigma_1,\dots,\lambda_n\sigma_n}\in \Le(\mathbb C^n,\Ha)^+
\end{equation}

\begin{lemma}\label{lemma:dual_guess}
Let $\Ee$ be an ensemble. Then  $P_{succ}(\Ee)=\|\phi_\Ee\|^\diamond$.
\end{lemma}

\begin{proof} Since $\phi_\Ee\ge 0$, we obtain by Proposition \ref{prop:properties} and a  similar reasoning  as in the proof of Lemma \ref{lemma:qccq},

\[
\|\phi_\Ee\|^\diamond=\sup_{\alpha\in \Ce(\Ha,\mathbb C^n)}\<\alpha,\phi_\Ee\>=\sup_{M\in \Me(\Ha,n)}\sum_i \Tr[M_i\lambda_i\sigma_i]=P_{succ}(\Ee).
\]

\end{proof}

 We next show that the dual norm of any completely positive map can be written as a (multiple of) the optimal success probability of some ensemble.

\begin{prop}\label{prop:dual_guess_all}
Let $\gamma\in \Le(\Ka,\Ha)^+$. Then there is an (equiprobable) ensemble $\Ee_\gamma$ on $\Ha\otimes \Ka$ such that 
\[
\|\gamma\|^\diamond=d_\Ka\Tr[\gamma(I)]P_{succ}(\Ee_\gamma).
\]
Moreover, for any $\phi\in \Le(\Ha,\Ha')^+$, we have 
\[
\Ee_{\phi\circ\gamma}=(\phi\otimes id)(\Ee_\gamma).
\]
\end{prop}

\begin{proof} Let  $\{U_1,\dots,U_{d_\Ka^2}\}$ be a set of Heisenberg-Weyl operators on $\Ka$, that is a set of  unitaries  such that 
\[
\frac 1{d_\Ka} \sum_j U_j^*AU_j=\Tr[A]I_\Ka,\qquad A\in B(\Ka).
\]
Define 
\[
\sigma^\gamma_i:=\Tr[\gamma(I)]^{-1}(I\otimes U_i^*)C(\gamma)(I\otimes U_i),\quad i=1,\dots, d_\Ka^2.
\]
Clearly, $\sigma^\gamma_i\ge 0$ and $\Tr[\sigma^\gamma_i]=\Tr[\gamma(I)]^{-1}\Tr[C(\gamma)]=1$, so that $\sigma_i^\gamma\in \states(\Ha\otimes \Ka)$. Let $\Ee_\gamma$ be  the (equiprobable) ensemble $\Ee_\gamma:=\{\frac 1{d_\Ka^2},\sigma^\gamma_i\}_{i=1}^{d_\Ka^2}$. By Lemma \ref{lemma:dual_guess}, 
\[
d_\Ka \Tr[\gamma(I)]P_{succ}(\Ee_\gamma)=d_\Ka \Tr[\gamma(I)]\|\phi_{\Ee_\gamma}\|^\diamond=\|\phi^{cq}_F\|^\diamond,
\]
where 
\[
F_i:=\frac 1{d_\Ka} \Tr[\gamma(I)]\sigma_i^\gamma=\frac 1{d_\Ka}(I\otimes U_i^*)C(\gamma)(I\otimes U_i),\quad i=1,\dots, 
d_\Ka^2
\]
 We will prove that $\|\gamma\|^\diamond=\|\phi^{cq}_F\|^\diamond$. Since $F_i\ge 0$, we have
 \begin{align*}
\|\phi^{cq}_F\|^\diamond=\sup_{M\in \Me(\Ha\otimes \Ka,d_\Ka^2)} \sum_i \Tr [M_iF_i].
\end{align*}
For any  $M\in \Me(\Ha\otimes \Ka,d_\Ka^2)$, 
\[
\sum_i\Tr [M_iF_i]=\Tr [\frac 1{d_\Ka}\sum_i(I\otimes U_i)M_i(I\otimes U_i^*)C(\gamma)]=\Tr[Y_MC(\gamma)],
\]
where $Y_M:=\frac 1{d_\Ka}\sum_i(I\otimes U_i)M_i(I\otimes U_i^*)$ is a positive operator on $\Ha\otimes \Ka$ such that
\[
\Tr_\Ka [Y_M]=\frac 1{d_\Ka}\Tr_\Ka [\sum_i M_i]=I_\Ha.
\]
It follows that $Y_M$ is the Choi matrix of some completely positive unital map in $\Le(\Ka,\Ha)$, \cite{choi1975completely}. Hence there is some channel $\phi\in \Ce(\Ha,\Ka)$ such that $Y_M=C(\phi^*)$ and by Remark \ref{rem:choi}
\[
\sum_i\Tr M_iF_i=\Tr C(\phi^*)C(\gamma)=\<\phi,\gamma\>.
\]
It follows that
\[
\|\phi^{cq}_F\|^\diamond\le \sup_{\phi\in \Ce(\Ha,\Ka)}\<\phi,\gamma\>=\|\gamma\|^\diamond.
\]
Conversely, let $\phi\in \Ce(\Ha,\Ka)$, then 
\begin{align*}
\<\phi,\gamma\>&=\Tr [C(\gamma)C(\phi^*)]=\frac 1{d_\Ka}\Tr[\sum_i F_i(I\otimes U^*_i)C(\phi^*)(I\otimes U_i)]\\
&=
\sum_i \Tr[F_iM_i],
\end{align*}
where $M_i=\frac 1{d_\Ka}(I\otimes U^*_i)C(\phi^*)(I\otimes U_i)$. Since $M_i\ge 0$ and
\[
\sum_i M_i=\Tr_\Ka C(\phi^*)\otimes I_\Ka=I_{\Ha\otimes \Ka},
\]
$M=\{M_1,\dots,M_{d_\Ha^2}\}$ is a POVM on $\Ha\otimes \Ka$. It follows that $\|\gamma\|^\diamond= \|\phi^{cq}_F\|^\diamond$.

\end{proof}

\section{The main result}

Let $\Phi\in \Ce(\Ha,\Ka)$ and $\Psi\in\Ce(\Ha,\Ka')$. Similarly to Le Cam's deficiency for statistical experiments, we may define 
the deficiency of $\Phi$ with respect to $\Psi$ by
\[
\delta(\Phi,\Psi)=\inf_{\alpha\in \Ce(\Ka',\Ka)}\|\Phi-\alpha\circ\Psi\|_\diamond.
\]
Since $\Ce(\Ka',\Ka)$ is convex and compact, the infimum is attained, in particular, $\delta(\Phi,\Psi)=0$ if and only if $\Phi=\alpha\circ\Psi$ for some $\alpha\in \Ce(\Ka',\Ka)$. In this case, we say that that $\Phi$ is a post-processing of $\Psi$ and 
write $\Phi\preceq \Psi$.  We also define  Le Cam distance by
\[
\Delta(\Phi,\Psi)=\max\{\delta(\Phi,\Psi),\delta(\Psi,\Phi)\}.
\]
This defines a preorder on the set of channels with the same input space.
The following data processing inequalities for $\delta$  are obvious consequences of the definition and Proposition \ref{prop:properties} (iii).
\begin{prop} Let $\Phi_1,\Phi_2,\Phi$, $\Psi_1,\Psi_2,\Psi$ be channels with the same input space. 
\begin{enumerate}
\item[(i)] If $\Phi_1\preceq \Phi_2$, then $\delta(\Phi_1,\Psi)\le\delta(\Phi_2,\Psi)$. 
\item[(ii)] If $\Psi_1\preceq \Psi_2$, then $\delta(\Phi,\Psi_1)\ge\delta(\Phi,\Psi_2)$.
\end{enumerate}

\end{prop}

Let now $\Ee$ be any ensemble on $\Ha$ and let 
$\Phi(\Ee)$ be the ensemble on $\Ka$ obtained by applying $\Phi$ to each state in $\Ee$. It follows from Lemma \ref{lemma:dual_guess} and Proposition \ref{prop:properties} (iii) (but is also easy to see directly), that if $\Phi$ is a post-processing of $\Psi$, we must have $P_{succ}(\Phi(\Ee))\le P_{succ}(\Psi(\Ee))$. In fact, for any ensemble $\Ee$ on the tensor product $\Ha\otimes \Ha_0$ with an ancillary Hilbert space $\Ha_0$, we have
\begin{equation}\label{eq:ineq}
P_{succ}((\Phi\otimes id_{\Ha_0})(\Ee))\le P_{succ}((\Psi\otimes id_{\Ha_0})(\Ee)).
\end{equation}
The converse was proved in  \cite{chefles2009quantum, buscemi2012comparison}, namely that if (\ref{eq:ineq}) holds for any ensemble and any ancilla, then $\Phi\preceq \Psi$. 

Our aim in the present section is to prove an  $\epsilon$-version of this result. 
More precisely, for  $\epsilon\ge 0$, we want to characterize  pairs of channels satisfying $\delta(\Phi,\Psi)\le \epsilon$  by comparing the  maximal success probabilities
for ensembles obtained by sending an ensemble through the two channels. We first prove a  "classical" variant, where no ancilla is present.  The important part of the following proposition is the equivalence of (i) and  (iv), relating the comparison of success probabilities to approximations of pre-processings of POVMs.  

\begin{prop}\label{prop:classical}
Let $\Phi\in \Ce(\Ha,\Ka)$ and $\Psi\in\Ce(\Ha,\Ka')$, $\epsilon \ge 0$, $k\in \mathbb N$. The following are equivalent.
\begin{enumerate}
\item[(i)] For any ensemble $\Ee=\{\lambda_i,\rho_i\}_{i=1}^k$ on $\Ha$, we have
\[
P_{succ}(\Phi(\Ee))\le P_{succ}(\Psi(\Ee))+ \epsilon P_{succ}(\Ee)
\]
\item[(ii)] For  any collection $F=\{F_1,\dots, F_k\}$  in $B(\Ha)^+$, we have
 \[
\|\Phi\circ\phi^{cq}_F\|^\diamond\le \|\Psi\circ\phi^{cq}_F\|^\diamond+\epsilon\|\phi^{cq}_F\|^\diamond
\]
\item[(iii)] For any  collection $F=\{F_1,\dots,F_k\}$ in $ B_h(\Ha)$, we have
\[
\max_{M\in \Me(\Ka,k)} \sum_i \Tr M_i\Phi(F_i)\le \max_{N\in \Me(\Ka',k)}
\sum_i \Tr N_i\Psi(F_i)+ 2\epsilon \|\phi^{cq}_F\|^\diamond
\]
\item[(iv)] For any $M\in \Me(\Ka,k)$, there is some $N\in \Me(\Ka',k)$ such that 
\[
\|\phi^{qc}_{\Phi^*(M)}-\phi^{qc}_{\Psi^*(N)}\|_\diamond\le 2\epsilon
\]

\end{enumerate}

\end{prop}

\begin{proof}
Suppose (i) and let $F_i\in  B(\Ha)^+$, $i=1,\dots,k$. Put  $c_i=\Tr F_i$, $c=\sum_ic_i$ and let $\Fe=\{\lambda_i,\sigma_i\}_{i=1}^k$, where  
$\sigma_i=(c_i)^{-1}F_i$ and  $\lambda_i=c_i/c$.
Then $\Fe$ is an ensemble on $\Ha$ and $\phi^{cq}_F=c\phi_\Fe$. By Lemma \ref{lemma:dual_guess}, we have 
\[
\|\Phi\circ \phi^{cq}_F\|^\diamond=c\|\Phi\circ \phi_\Fe\|^\diamond=c\|\phi_{\Phi(\Fe)}\|^\diamond=cP_{succ}(\Phi(\Fe))
\] 
and similarly for $\Psi$, this proves (ii). 

Suppose (ii) and let $F_i \in B_h(\Ha)$, $i=1,\dots,k$.  Let $\|\phi^{cq}_F\|^\diamond=t$, then by (\ref{eq:updiamondinf}) there is some $\sigma\in \states(\Ha)$ such that $2t\phi_\sigma\ge \phi^{cq}_F+t\phi_\sigma\ge 0$, so that $F_i+t\sigma\in B(\Ha)^+$ for all $i$. By (ii), we have for any $M\in \Me(\Ha,k)$
\begin{align*}
\sum_i \Tr M_i\Phi(F_i)+t&=\sum_i \Tr M_i\Phi(F_i+t\sigma)=\<\phi^{qc}_M,\Phi\circ\phi^{cq}_{F+t\sigma}\>\\
&\le 
\|\Phi\circ(\phi^{cq}_{F+t\sigma})\|^\diamond\le \|\Psi\circ(\phi^{cq}_{F+t\sigma})\|^\diamond+\epsilon \|\phi^{cq}_{F+t\sigma}\|^\diamond\\
&=\sup_{N\in \Me(\Ka',k)}\sum_i\Tr N_i\Psi(F_i) +t+\epsilon \|\phi^{cq}_{F+t\sigma}\|^\diamond
\end{align*}
here $F+t\sigma=\{F_1+t\sigma,\dots,F_k+t\sigma\}$.   It is now enough to notice that $0\le \phi^{cq}_{F+t\sigma}\le 2t\phi_\sigma$, so that $\|\phi^{cq}_{F+t\sigma}\|^\diamond\le 2t=2\|\phi^{cq}_F\|^\diamond$.

Next, suppose (iii) and let $M\in \Me(\Ha,k)$. Then  
\begin{align*}
2\epsilon&\ge \max_{\|\phi^{cq}_F\|^\diamond\le 1}\min_{N\in \Me(\Ka',k)} \left(\sum_i\Tr\Phi(F_i)M_i-\sum_i\Tr \Psi(F_i)N_i\right)\\
&=\max_{\|\phi^{cq}_F\|^\diamond\le 1}\min_{N\in \Me(\Ka',k)}\<\phi^{qc}_{\Phi^*(M)}-\phi^{qc}_{\Psi^*(N)},\phi^{cq}_F\>\\
&=\max_{\|\psi\|^\diamond\le 1}\min_{N\in \Me(\Ka',k)}\<\phi^{qc}_{\Phi^*(M)}-\phi^{qc}_{\Psi^*(N)},\psi\>.
\end{align*}
Since $\Me(\Ka',k)$ and the unit ball of $\|\cdot\|^\diamond$ are compact convex sets and the map $(\psi,N)\mapsto \<\phi^{qc}_{\Phi^*(M)}-\phi^{qc}_{\Psi^*(N)},\psi\>$ is linear in both arguments, the minimax theorem (see e.g. \cite{strasser1985statistics}) applies and we have
\begin{align*}
\min_{N\in \Me(\Ka',k)} \|\phi^{qc}_{\Phi^*(M)}-\phi^{qc}_{\Psi^*(N)}\|_\diamond&=\min_{N\in \Me(\Ka',k)}\max_{\|\psi\|^\diamond\le1}\<\phi^{qc}_{\Phi^*(M)}-\phi^{qc}_{\Psi^*(N)},\psi\>\\
&=\max_{\|\psi\|^\diamond\le 1}\min_{N\in \Me(\Ka',k)}\<\phi^{qc}_{\Phi^*(M)}-\phi^{qc}_{\Psi^*(N)},\psi\>\\
&\le 2\epsilon.
\end{align*}

Finally, suppose (iv) and let $\Ee=\{\lambda_i,\sigma_i\}_{i=1}^k$ be an ensemble on $\Ha$. Let $M\in \Me(\Ka',k)$ and let $N\in \Me(\Ka',k)$ be as in (iv). 
Then 
\begin{align*}
\sum_i \Tr M_i \lambda_i\Phi(\sigma_i)&=\<\phi^{qc}_{\Phi^*(M)},\phi_\Ee\>\le 
\<\phi^{qc}_{\Psi^*(N)},\phi_\Ee\>+\frac12\|\phi^{qc}_{\Phi^*(M)}-\phi^{qc}_{\Psi^*(N)}\|_\diamond\|\phi_\Ee\|^\diamond\\
&\le \sum_i\Tr N_i \lambda_i \Psi(\sigma_i)+\epsilon \|\phi_\Ee\|^\diamond.
\end{align*}
By Lemma \ref{lemma:dual_guess}, we obtain (i).

\end{proof}

\begin{rem} Assume that  $\Phi$ in the above proposition is a qc-channel. In this case, $\delta_\Ka\circ \Phi=\Phi$ and by putting 
$M=\{|e_1^\Ka\>\<e_1^\Ka|,\dots,|e_{d_\Ka}^\Ka\>\<e_{d_\Ka}^\Ka|   \}$ in (iv), we obtain that (i) is equivalent to $\delta(\Phi,\Psi)\le 2\epsilon$. This also gives a randomization criterion for classical-to-classical channels. Compare this to \cite{buscemi2015degradable}, where an analogous statement for such channels is proved but with a dimension-dependent factor multiplying $\epsilon$.

\end{rem}

We now prove our main result.

\begin{thm}\label{thm:rand_all} Let $\Phi\in \Ce(\Ha,\Ka)$, $\Psi\in \Ce(\Ha,\Ka')$, $\epsilon\ge 0$. The following statements are equivalent.
\begin{enumerate}
\item[(i)] $\delta(\Phi,\Psi)\le \epsilon$.

\item[(ii)] For any finite dimensional Hilbert space $\Ka_0$ and any map $\gamma\in \Le(\Ka_0,\Ha)^+$, 
\[
\|\Phi\circ\gamma\|^\diamond\le \|\Psi\circ\gamma\|^\diamond+\frac{\epsilon}2\|\gamma\|^\diamond.
\]
\item[(iii)] For any finite dimensional Hilbert space $\Ka_0$ and any ensemble $\Ee=\{\lambda_i,\sigma_i\}_{i=1}^k$ on $\Ha\otimes \Ka_0$,  
\[
P_{succ}((\Phi\otimes id_{\Ka_0})(\Ee))\le P_{succ}((\Psi\otimes id_{\Ka_0})(\Ee))+\frac{\epsilon}2P_{succ}(\Ee).
\]
\end{enumerate}
Moreover, in (ii) and (iii), one can restrict to $\Ka_0=\Ka$ and equiprobable ensembles with $k=d_\Ka^2$ elements.

\end{thm}

\begin{proof} Assume (i) and let $\alpha\in \Ce(\Ka',\Ka)$ be a channel such that
\[
\|\Phi-\alpha\circ\Psi\|_\diamond\le \epsilon.
\]
 Then for any $\gamma\in \Le(\Ka_0,\Ha)^+$ and $\chi\in \Ce(\Ka,\Ka_0)$, we have by positivity and Proposition \ref{prop:properties} (ii), (iii) that
\begin{align*}
\<\gamma,\chi\circ\Phi\>&\le\<\gamma,\chi\circ\alpha\circ\Psi\>+|\<\gamma,\chi\circ(\Phi-\alpha\circ\Psi)\>|\\
&\le \<\gamma,\chi\circ\alpha\circ\Psi\>+\frac12\|\gamma\|^\diamond\|\chi\circ(\alpha\circ\Psi-\Phi)\|_\diamond\\
&\le \<\gamma,\chi\circ\alpha\circ\Psi\>+\frac12\epsilon\|\gamma\|^\diamond
\end{align*}
By Proposition \ref{prop:properties} (i) and properties of $s$, we have
\begin{align*}
\|\Phi\circ\gamma\|^\diamond&=\sup_{\chi\in \Ce(\Ka,\Ka_0)}\<\chi,\Phi\circ\gamma\>=\sup_{\chi\in \Ce(\Ka,\Ka_0)}\<\gamma,\chi\circ\Phi\>\\
&\le \sup_{\chi\in \Ce(\Ka,\Ka_0)}\<\gamma,\chi\circ\alpha\circ\Psi\>+\frac12\epsilon\|\gamma\|^\diamond\\&=\sup_{\chi\in \Ce(\Ka,\Ka_0)}\<\chi\circ\alpha,\Psi\circ\gamma\>+\frac12\epsilon\|\gamma\|^\diamond\\
&\le \sup_{\xi\in \Ce(\Ka',\Ka_0)}\<\xi,\Psi\circ\gamma\>+\frac12\epsilon\|\gamma\|^\diamond=\|\Psi\circ\gamma\|^\diamond +\frac12\epsilon\|\gamma\|^\diamond.
\end{align*}
Hence (i) implies (ii). We will now  prove the converse.
So suppose (ii), with $\Ka_0=\Ka$. By Proposition \ref{prop:properties} (ii), we have
\[
\min_{\alpha\in \Ce(\Ka',\Ka)}\|\Phi-\alpha\circ\Psi\|_\diamond=
2\min_{\alpha\in \Ce(\Ka',\Ka)}\left\{\max_{\substack{\gamma\in \Le(\Ka,\Ha)^+,\\ \|\gamma\|^\diamond\le 1}} \<\gamma,\Phi-\alpha\circ\Psi\>\right\}
\]  
Similarly as in the proof of Proposition \ref{prop:classical}, we may apply the minimax theorem. We obtain 
\begin{align*}
\min_{\alpha\in \Ce(\Ka',\Ka)}\|\Phi-\alpha\circ\Psi\|_\diamond&=2\max_\gamma\min_\alpha\<\gamma,\Phi-\alpha\circ\Psi\>\\
&=2\max_\gamma\left\{ \<\gamma,\Phi\>-\|\Psi\circ\gamma\|^\diamond\right\}\\
&\le 2\max_\gamma \left\{\|\Phi\circ\gamma\|^\diamond-\|\Psi\circ\gamma\|^\diamond\right\}\le \epsilon.
\end{align*}
Hence (i) and (ii) are equivalent, moreover, it is enough to assume $\Ka_0=\Ka$ in (ii). 

Next, suppose (i). Since by (\ref{eq:diamond}), we have $\|\phi\|_\diamond=\|\phi\otimes id_{\Ka_0}\|_\diamond$ for any $\Ka_0$, we have $\delta(\Phi\otimes id_{\Ka_0},\Psi\otimes id_{\Ka_0})\le \epsilon$. As we just proved above, it follows that for any
$\Ka_1$ and any map $\gamma\in \Le(\Ka_1,\Ha\otimes\Ka_0)^+$, we have 
\[
\|(\Phi\otimes id_{\Ka_0})\circ\gamma\|^\diamond\le \|(\Psi\otimes id_{\Ka_0})\circ\gamma\|^\diamond+\frac{\epsilon}2\|\gamma\|^\diamond.
\]
If $\Ee$ is any ensemble on $\Ha\otimes\Ka_0$, then by putting $\gamma=\phi_\Ee$ and using Lemma \ref{lemma:dual_guess} we obtain the inequality in (iii). Now it is enough to use Proposition \ref{prop:dual_guess_all} to obtain the implication (iii) $\implies$ (ii).

\end{proof}

The following two results were already obtained in \cite{buscemi2012comparison}.
In particular, Corollary \ref{coro:rand_0} shows that for $\epsilon=0$ one can restrict to ensembles of separable states. We will give the proofs in our setting.

\begin{thm}\label{thm:rand_0} Let $\Phi\in \Ce(\Ha,\Ka)$, $\Psi\in \Ce(\Ha,\Ka')$ and let $\xi\in \Ce(\Ka_0,\Ka)$ be a surjective channel.
 Then $\Phi\preceq \Psi$ if and only if for any ensemble $\Ee$ on $\Ha\otimes \Ka_0$,
\[
P_{succ}((\Phi\otimes\xi)(\Ee))\le P_{succ}((\Psi\otimes\xi)(\Ee)).
\]

\end{thm}

\begin{proof} The 'if' part of the theorem follows by the monotonicity property of the optimal success probabilities. For the converse, let $\Ee=\{\lambda_i,\sigma_i\}_{i=1}^k$ be any ensemble on $\Ha\otimes \Ka$. We will show that 
\[
P_{succ}((\Phi\otimes id)(\Ee))\le P_{succ}((\Psi\otimes id)(\Ee))
\]
which implies the statement by Theorem \ref{thm:rand_all}. Since $\xi$ is surjective, and hence $id\otimes \xi$ must be surjective as well, there are some 
$F_1,\dots,F_k$ in $B(\Ha\otimes \Ka_0)$ such that $(id\otimes \xi)(F_i)=\lambda_i\sigma_i$. Obviously, we may assume that $F_i=F_i^*$. Then by Proposition \ref{prop:classical} (iii), 
\begin{align*}
P_{succ}((\Phi\otimes id)(\Ee))&=\sup_{M\in \Me(\Ka\otimes \Ka,k)}\sum_i \Tr M_i(\Phi\otimes id)(\lambda_i\sigma_i)\\
&=\sup_{M\in \Me(\Ka\otimes \Ka,k)}\sum_i \Tr M_i(\Phi\otimes \xi)(F_i)\\
 &\le 
\sup_{N\in \Me(\Ka'\otimes \Ka,k)}\sum_i \Tr N_i(\Psi\otimes \xi)(F_i)\\
&=
\sup_{N\in \Me(\Ka'\otimes \Ka,k)}\sum_i\Tr N_i(\Psi\otimes id)(\lambda_i\sigma_i)\\
&=P_{succ}((\Psi\otimes id)(\Ee))
\end{align*}

\end{proof}

\begin{coro}\label{coro:rand_0} Let $\Phi\in \Ce(\Ha,\Ka)$, $\Psi\in \Ce(\Ha,\Ka')$. Let $\Se=\{\sigma_j\}\subset \states(\Ka)$  be a finite subset that spans  $B_h(\Ka)$. Then $\Phi\preceq \Psi$  if and only if for any ensemble $\Ee=\{\lambda_i,\rho_i\}_{i=1}^m$ on $\Ha\otimes \Ka$ of states of the form  $\rho_i=\sum_j \rho^i_j\otimes \sigma_j$ with $\rho^i_j\in B(\Ha)^+$, $\sum_j\Tr \rho^i_j=1$,  we have
\[
P_{succ}((\Phi\otimes id)(\Ee))\le P_{succ}((\Psi\otimes id)(\Ee)).
\]
Again, it is enough to assume $m=d_\Ka^2$ and $\lambda_i=\lambda_j$ for all $i,j$.

\end{coro}

\begin{proof} We use Theorem \ref{thm:rand_0}. It is clear that $\xi=\phi^{cq}_{\Se}$ is a surjective channel and that for any ensemble $\Ee$ on $\Ha\otimes \Ka_0$, $(id\otimes \xi)(\Ee)$ has the above form.

\end{proof}

\section{The randomization criterion for quantum experiments}

A quantum statistical experiment, or just an experiment,  is a pair $\Te=(\Ha,\{\rho_\theta,\ \theta\in \Theta\})$, where  $\rho_\theta\in \states(\Ha)$ for all $\theta\in \Theta$ and $\Theta$ is an arbitrary set of parameters. Any experiment can be 
viewed as the set of possible states of some physical system, determined by some prior information on the true state. Note that this definition contains also classical statistical experiments on finite sample spaces, which can be identified with diagonal density matrices.

 Based on the outcome of a measurement on the system, a decision $j$  is chosen from a (finite) set $D$ of decisions. This procedure, or a decision rule, is represented by a POVM on $\Ha$ with outcomes in $D$. The performance of a decision rule is assessed by a {payoff function}, which in our case is a map $g: \Theta\times D\to \mathbb R^+$,  representing the {payoff} 
obtained if $j\in D$ is chosen while the true state is $\rho_\theta$.  The average payoff of the decision rule $M$ at $\theta\in \Theta$ is computed as 
\[
P_\Te(\theta,M,g)=\sum_{j\in D} g(\theta,j)\Tr\rho_\theta M_j.
\]
We call the pair $(D,g)$ a {decision space}.  

The theory of classical statistical experiments and their comparison was introduced by Blackwell in \cite{blackwell1951comparison} and further developed by Le Cam \cite{lecam1964sufficiency} and many other authors, see e.g. \cite{lecam1986asymptotic, torgersen1991comparison, strasser1985statistics} for more information. Following the classical definition, we may introduce the notion of (classical) deficiency for quantum experiments.

\begin{defi}\label{def:defi_exps}
Let $\Se=(\Ka, \{\sigma_\theta,\ \theta\in \Theta\})$ and $\Te=(\Ha,\{\rho_\theta,\ \theta\in\Theta\})$  be quantum statistical  experiments and let $\epsilon \ge 0$. We say that $\Te$ is classically { $\epsilon$-deficient} relative
 to $\Se$, in notation $\Se\preceq_{cl,\epsilon}\Te$, if for any  decision space $(D,g)$  and any $M\in \Me(\Ka,|D|)$, there is some $N\in \Me( \Ha,|D|) $ such that
\begin{equation}\label{eq:defi}
\sup_{\theta\in \Theta} \left[P_\Se(\theta,M,g)- P_\Te(\theta,N,g)-\epsilon\max_dg(\theta,d)\right]\le 0.
\end{equation}

\end{defi}

\begin{rem}\label{rem:guessing} Note that Definition \ref{def:defi_exps} can be rewritten  in terms of guessing probabilities. 
Indeed, let $\Pe_\Theta$ be the set of probability measures on $\Theta$ with finite support. By \cite[Theorem 3]{lecam1964sufficiency}, $\Se\preceq_{cl,\epsilon}\Te$ if and only if for any  decision space $(D,g)$ and 
any $p\in \Pe_\Theta$, we have
\begin{equation}\label{eq:deficiencyprob}
\inf_{M} P_\Se(p,M,g)\le \inf_{N} P_\Te(p,N,g)+\epsilon \sum_\theta \max_dp(\theta)g(\theta,d)
\end{equation}
where $P_\Se(p,M,g)=\sum_\theta p(\theta)P_\Se(\theta,M,g)$. The last inequality can be easily rewritten as 
\begin{equation}\label{eq:deficiencyensemble}
P_{succ}(\{\lambda_d, \sum_\theta \mu^d_\theta\sigma_\theta\})\le P_{succ}(\{\lambda_d, \sum_\theta \mu^d_\theta\rho_\theta\})+\epsilon P_{succ}(\Ee)
\end{equation}
 for some classical ensemble $\Ee=\{\lambda_d, \mathrm{diag}(\mu^d_\theta, \theta\in \mathrm{supp}(p))\}$.
 Conversely, for any such ensemble $\Ee$ one can find some $p\in \Pe_\Theta$ and a decision space $(D,g)$ such that $(\ref{eq:deficiencyprob})$ is equivalent to (\ref{eq:deficiencyensemble}). 

\end{rem}

 The next theorem is the celebrated Le Cam's randomization criterion for classical experiments. Note that our setting contains only experiments on finite sample spaces, but the theorem holds in a much more general case. 

\begin{thm}\label{thm:rc}\cite{lecam1964sufficiency} Let $\Te=(\Ha,\{\rho_\theta,\ \theta\in\Theta\})$ and $\Se=(\Ka, \{\sigma_\theta,\ \theta\in \Theta\})$ be classical statistical experiments. Then $\Se\preceq_{cl,\epsilon}\Te$ if and only if there is some channel $\alpha$ such that 
\[
\sup_\theta\|\sigma_\theta-\alpha(\rho_\theta)\|_1\le 2\epsilon.
\]
\end{thm}

Classical deficiency of quantum experiments was studied in \cite{jencova2012comparison,matsumoto2010randomization}. It was proved in \cite{matsumoto2014anexample} that Theorem \ref{thm:rc} is not true for quantum experiments, even if we assume that $\epsilon =0$ and that $\Te$ is classical.  Matsumoto in \cite{matsumoto2010randomization} proved a randomization criterion in terms of quantum decision spaces, defined as  pairs $(\De,G)$, consisting of a Hilbert space $\De$ and a payoff map $G:\Theta \to B(\De)^+$.  In this case, decision rules  are represented by channels 
$\phi\in \Ce(\Ha,\De)$ and the payoff is computed as 
\[
P_\Te(\theta,\phi,G)=\Tr \phi(\rho_\theta)G(\theta).
\] 
Quantum $\epsilon$-deficiency, denoted by $\Se\preceq_\epsilon \Te$, is then defined analogically as in Definition \ref{def:defi_exps}. An operational interpretation of quantum decision spaces  is not clear.

The aim of the present section is to obtain a quantum version of the randomization criterion using Theorem \ref{thm:rand_all}. Our version is based on comparison of guessing probabilities for certain ensembles obtained from the two experiments and by Remark \ref{rem:guessing} 
can be seen as an extension of Le Cam's theorem. We also include the proof of Matsumoto's randomization criterion, which fits nicely into our framework.  
  
\begin{thm}\label{thm:rand_crit} Let $\Se=(\Ka, \{\sigma_\theta,\ \theta\in \Theta\})$ and $\Te=(\Ha,\{\rho_\theta,\ \theta\in\Theta\})$  be quantum statistical  experiments and let $\epsilon \ge 0$. Then the following are equivalent.
\begin{enumerate}
\item[(i)] $\Se\preceq_{\epsilon}\Te$
\item[(ii)] For any finite subset $\{\theta_1,\dots,\theta_n\}\subseteq \Theta$ and any ensemble $\Ee=\{\lambda_i,\tau_i\}_{i=1}^k$  on $\mathbb C^{n}\otimes \Ka$, consisting of block-diagonal states $\tau_i=\sum_{j=1}^n|e^{n}_j\>\<e^{n}_j|\otimes \tau_i^j$, $\tau_i^j\in B(\Ka)^+$, we have
\[
P_{succ}(\{\lambda_i,\sum_{j=1}^n\sigma_{\theta_j}\otimes \tau^j_{i}\})\le P_{succ}(\{\lambda_i,\sum_{j=1}^n\rho_{\theta_j}\otimes \tau^j_{i}\})+\epsilon P_{succ}(\Ee).
\]
\item[(iii)] There is some $\alpha\in \Ce(\Ha,\Ka)$ such that 
\[
\sup_{\theta\in \Theta} \|\sigma_\theta-\alpha(\rho_\theta)\|_1\le 2\epsilon.
\]

\end{enumerate}
Moreover, we may restrict to equiprobable ensembles with $k=d_\Ka^2$.

\end{thm}

\begin{proof} Suppose (i) and let $\{\theta_1,\dots,\theta_n\}\subseteq \Theta$. Let $\De=\Ka$ and let  $G:\Theta\to B(\Ka)^+$ be such that  $G(\theta)=0$ outside $\{\theta_1,\dots,\theta_n\}$. Then by (i), for any $\phi\in \Ce(\Ka,\Ka)$  there is some 
$\phi'\in \Ce(\Ha,\Ka)$ such that 
\begin{equation}\label{eq:deficiency_finite}
\sum_{j=1}^n\left[P_\Se(\theta_j,\phi ,G)- P_\Te(\theta_j,\phi',G)-\epsilon\|G(\theta_j)\|\right]\le 0.
\end{equation}
We may identify $G$ with the collection of operators $\{G(\theta_j),\ j=1,\dots,n\}$. Let also $\Se_0=\{\sigma_{\theta_1},\dots,\sigma_{\theta_n}\}$ and $\Te_0=\{\rho_{\theta_1},\dots,\rho_{\theta_n}\}$. Using Lemma \ref{lemma:qccq}, it is easy to see that (\ref{eq:deficiency_finite}) can be written as
\begin{equation}\label{eq:deficiency_finite_chan}
\<\phi^{qc}_G,\phi\circ\phi^{cq}_{\Se_0}\>\le \<\phi^{qc}_G,\phi'\circ\phi^{cq}_{\Te_0}\>+\epsilon \|\phi^{qc}_G\|^\diamond.
\end{equation}
Using  properties of the duality  $\<\cdot,\cdot\>$ and Proposition \ref{prop:properties}, we obtain for any $\phi\in \Ce(\Ka,\Ka)$,
\[
\<\phi,\phi^{cq}_{\Se_0}\circ\phi^{qc}_G\>\le \sup_{\phi'\in \Ce(\Ha,\Ka)}
\<\phi',\phi^{cq}_{\Te_0}\circ\phi^{qc}_G\>+\epsilon \|\phi^{qc}_G\|^\diamond=\|\phi^{cq}_{\Te_0}\circ\phi^{qc}_G\|^\diamond +\epsilon \|\phi^{qc}_G\|^\diamond,
\]
hence
\[
\|\phi^{cq}_{\Se_0}\circ\phi^{qc}_G\|^\diamond\le \|\phi^{cq}_{\Te_0}\circ\phi^{qc}_G\|^\diamond +\epsilon \|\phi^{qc}_G\|^\diamond
\]
holds for any completely positive qc-map $\phi^{qc}_G$. Since $\phi^{cq}_{\Se_0}\circ\gamma=\phi^{cq}_{\Se_0}\circ\delta_{\mathbb C^{n}}\circ\gamma$, similarly for $\phi^{cq}_{\Te_0}$, and since $\delta_{\mathbb C^{n}}\circ\gamma$ is a qc-map with 
$\|\delta_{\mathbb C^{n}}\circ\gamma\|^\diamond\le \|\gamma\|^\diamond$, we may replace $\phi^{qc}_G$ by any map $\gamma\in \Le(\Ka,\mathbb C^n)$ and obtain
\begin{equation}\label{eq:exps_ii}
\|\phi^{cq}_{\Se_0}\circ\gamma\|^\diamond\le \|\phi^{cq}_{\Te_0}\circ\gamma\|^\diamond +\epsilon \|\gamma\|^\diamond.
\end{equation}
Let now $\Ee$ be as in (ii),  then 
\[
(\phi^{cq}_{\Se_0}\otimes id)(\tau_i)=\sum_{j=1}^n\sigma_j\otimes \tau^j_{i}
\]
and similarly for $\phi^{cq}_{\Te_0}$, so that (ii) follows from (\ref{eq:exps_ii}) and Theorem \ref{thm:rand_crit}. We have proved that (i) implies (ii).

Next, suppose (ii) and let $\{\theta_1,\dots,\theta_n\}\subseteq \Theta$. Let $\Ee$ be any ensemble on $\mathbb C^n\otimes \Ka$. Then $\Ee':=(\delta_{\mathbb C^n}\otimes id)(\Ee)$ is an ensemble of block-diagonal states as in (ii) and 
$P_{succ}(\Ee')\le P_{succ}(\Ee)$, moreover, $(\psi\otimes id)(\Ee)=(\psi\otimes id)(\Ee')$ holds for any cq-channel $\psi$. Hence (ii) implies that
\[
P_{succ}((\phi^{cq}_{\Se_0}\otimes id)(\Ee))\le P_{succ}((\phi^{cq}_{\Se_0}\otimes id)(\Ee))+\epsilon P_{succ}(\Ee).
\]
By Theorem \ref{thm:rand_crit} and Lemma \ref{lemma:qccq}, this is equivalent to
\[
\min_{\alpha\in \Ce(\Ha,\Ka)}\sup_{j}\|\sigma_{\theta_j}-\alpha(\rho_{\theta_j})\|_1\le 2\epsilon.
\]
This clearly implies 
\[
\min_{\alpha\in \Ce(\Ha,\Ka)}\sum_{\theta}p(\theta)\|\sigma_\theta-\alpha(\rho_\theta)\|_1\le 2\epsilon
\]
for any $p\in \Pe_\Theta$ with support in $\{\theta_1,\dots,\theta_n\}$. Since this holds for any finite subset, 
\[
\sup_{p\in \Pe_\Theta}\min_{\alpha\in \Ce(\Ha,\Ka)}\sum_{\theta\in \Theta}p(\theta)\|\sigma_\theta-\alpha(\rho_\theta)\|_1\le 2\epsilon.
\]
Now we use the minimax theorem once more. For this, note that $\Pe_\Theta$ is a convex set, $\Ce(\Ha,\Ka)$ is compact and convex and the function $p\mapsto \sum_{\theta}p(\theta)\|\sigma_\theta-\alpha(\rho_\theta)\|_1$ is linear in $p$. It is also not difficult to see that the map $\alpha\mapsto  \|\sigma_\theta-\alpha(\rho_\theta)\|_1$ is continuous in the diamond norm. The minimax theorem can be applied and we obtain
\begin{align*}
\sup_{p\in \Pe_\Theta}\min_{\alpha\in \Ce(\Ha,\Ka)}\sum_{\theta}p(\theta)\|\sigma_\theta-\phi(\rho_\theta)\|_1&=\min_{\alpha\in \Ce(\Ha,\Ka)}\sup_{p\in \Pe_\Theta}\sum_{\theta}p(\theta)\|\sigma_\theta-\alpha(\rho_\theta)\|_1\\
&=\min_{\alpha\in \Ce(\Ha,\Ka)}\sup_{\theta}\|\sigma_\theta-\alpha(\rho_\theta)\|_1.
\end{align*}
Hence (ii) implies (iii).

Finally, suppose (iii) and let $(\De,G)$ be any quantum decision problem. Let $\phi\in \Ce(\Ka,\De)$ and put  $\phi'=\phi\circ\alpha$. Then for any $\theta\in \Theta$,
\begin{align*}
\Tr \phi(\sigma_\theta)G(\theta)-\Tr \phi'(\rho_\theta)G(\theta)&=\Tr \phi(\sigma_\theta-\alpha(\rho_\theta))G(\theta)\le \epsilon \|G(\theta)\|,
\end{align*}
which implies (i).

\end{proof}

Let $\Se$ and $\Te$ be experiments such that $\Se\preceq_0 \Te$. By the previous theorem, this is equivalent to existence of a channel $\alpha$ such that $\sigma_\theta=\alpha(\rho_\theta)$ for  $\theta\in \Theta$. In this case, we say that $\Se$ is a randomization of $\Te$.  An application of Theorem \ref{thm:rand_0} to  quantum statistical experiments shows that in this case, we may restrict to a special type of ensembles.

Below, we will say that an experiment $\Se_0=(\Ka, \{\tau_\theta, \theta\in \Theta_0\})$ is complete if the set $\{\tau_\theta, \theta\in \Theta_0\}$ spans $B(\Ka)$. If $\Theta_0$ is a finite set, then $\phi^{cq}_{\Se_0}=\phi^{cq}_{\{\tau_\theta, \theta\in \Theta_0\}}$ is a surjective channel.

\begin{coro}\label{coro:qrc_0} Let $\Se=(\Ka, \{\sigma_\theta,\ \theta\in \Theta\})$ and $\Te=(\Ha,\{\rho_\theta,\ \theta\in\Theta\})$  be quantum statistical  experiments. Let $\Se_0=(\Ka, \{\tau_1,\dots,\tau_N\})$ be a complete experiment. Then $\Se$ is a randomization of $\Te$ if and only if for any 
$\{\theta_1,\dots,\theta_n\}\subseteq \Theta$ and any collection $\{\Lambda_{j,l}^i\}$, $i=1,\dots,k$, $j=1,\dots, N$, $l=1,\dots,n$ of nonnegative numbers such that $\sum_{j,l}\Lambda^i_{j,l}=1$ for all $i$,
 we have
\[
P_{succ}(\{1/k, \sum_{j,l}\Lambda^i_{j,l}\sigma_{\theta_l}\otimes \tau_j\})\le P_{succ}(\{1/k, \sum_{j,l}\Lambda^i_{j,l}\rho_{\theta_l}\otimes \tau_j\}).
\] 

\end{coro}

Let now $\Phi\in \Ce(\Ha,\Ka)$ and $\Psi\in \Ce(\Ha,\Ka')$ and let $\Te_0$ be  a complete experiment on $\Ha$. 
Then it is clear that $\Phi\preceq \Psi$ if and only if $\Phi(\Te_0)$ is a randomization of $\Psi(\Te_0)$. The next result 
follows easily  by an application of Corollary \ref{coro:qrc_0}.

\begin{coro} Let $\Phi\in \Ce(\Ha,\Ka)$, $\Psi\in \Ce(\Ha,\Ka')$. Let $\Te_0=(\Ha,\{\tau^\Ha_1,\dots,\tau^\Ha_{M}\})$, 
$\Se_0=(\Ka,\{\tau^\Ka_1,\dots,\tau^\Ka_N\})$ be complete experiments. Then $\delta(\Phi,\Psi)=0$ if and only if 
\[
P_{succ}((\Phi\otimes id_{\Ka_0})(\Ee))\le P_{succ}((\Psi\otimes id_{\Ka_0})(\Ee))
\]

holds for all ensembles of states of the form 
\[
\Ee=\{\lambda_i, \sum_{j,l}\Lambda^i_{j,l}\tau_l^\Ha\otimes \tau^\Ka_j\}.
\]

\end{coro}

\section{Concluding remarks}

We have proved a version of the randomization criterion for quantum channels and applied the results to obtain a randomization criterion for quantum experiments. The deficiency $\delta(\Phi,\Psi)$ appears naturally in quantum information theory, for example in the definition of the approximately (anti)degradable channels \cite{suscre2015approximate}. Our results may be used to give an operational definition, similarly as it was done for antidegradable channels in \cite{budast2014game}.  Another possible application is to $\epsilon$-private and $\epsilon$-correctable channels \cite{kks2008complementarity}. 

Our proofs used properties of the diamond norm and its dual that can be obtained solely from the order structure given by completely positive maps and the trace preserving condition. This suggests the possibility to apply similar methods to more general situations, e. g. when complete positivity is replaced by other positivity assumptions, or to  more specific quantum protocols. 
One can also consider more general kinds of processings. Instead of post-processings where only the output is processed, it is possible to find conditions for approximation of one channel by pre-processings of the other. In the special case of qc-channels, or POVMs, this means that 
one POVM is $\epsilon$-cleaner than the other, which is an approximate version of the ordering of POVMs by cleanness, introduced in \cite{bkdpw2005clean}. The corresponding randomization criterion will be investigated in a forthcoming paper, see also \cite{jencova2015randomization}. 

More generally, the processing can consist of a combination of pre- and post-processing, also allowing some correlations between the input and the output, either classical or quantum. This would be  closer to the original classical definition by Shannon, \cite{shannon1958anote}. It seems that all these situations can be treated within the suggested framework.

Following the work of Raginsky \cite{raginsky2011shannon}, one can also define a deficiency measure based on more general distances, such as some generalized relative entropies.

Another challenging problem is the extension of these results to infinite dimensional Hilbert spaces, or to general von Neumann algebras. Some partial results in this direction were obtained in \cite{kalulu2013quantum}. Although the methods used in \cite{jencova2014base, jencova2015randomization} rely on finite dimensions, it seems plausible that some of the useful properties of the norms can be extended also to this case. All these problems are left for future work.

\section*{Acknowledgement}

This work was supported by the grants  VEGA 2/0069/16 and by Science and Technology Assistance Agency under the contract no. APVV-0178-11.


\end{document}